\documentclass[authoryear,preprint]{elsarticle}

\usepackage{amssymb}
\usepackage{amsmath}
\usepackage{amsthm}
\newtheorem{prop}{Proposition}
\newtheorem*{remark}{Remark}

\usepackage{bm} 
\usepackage{optidef} 
\usepackage{empheq}
\usepackage{algorithm}
\usepackage{algpseudocode}
\usepackage{algorithmicx}
\algnewcommand\algorithmicinput{\textbf{Step 1:}}
\algnewcommand\INPUT{\item[\algorithmicinput]}

\usepackage{booktabs} 
\usepackage{siunitx}
\usepackage{graphicx}
\usepackage{epstopdf}
\usepackage{subcaption}
\usepackage{xurl}
\usepackage{hyperref}          

\usepackage{moresize} 

\journal{journal (under review)}

\begin{document}

\begin{frontmatter}

\title{Dynamical System-Based Imitation Learning and Neuroadaptive Control for Trajectory Recovery in Autonomous Ships}

\author[a,b]{Yeyson A. Becerra-Mora\corref{cor1}}
\ead{ybecerra@us.es}
\author[a]{Jos\'e \'Angel Acosta}
\ead{jaar@us.es} 

\affiliation[a]{organization={Dept. Ingenieria de Sistemas y Automatica, University of Seville},
city={Sevilla},
postcode={41092},
country={Spain}}

\affiliation[b]{organization={Dept. of Electronic Engineering, CUN},
city={Bogota},
postcode={111711},
country={Colombia}}

\begin{abstract}
Repetitive maritime operations can be effectively learned using the Imitation Learning (IL) paradigm, which transfers human expertise directly to Unmanned Surface Vehicle (USV) control systems. Dynamical Systems (DS) are widely used to model non-linear human demonstrations while offering inherent stability guarantees. However, real-world execution under persistent marine perturbations reveals a critical trade-off: standard DS-based IL approaches prioritize global target convergence at the expense of localized trajectory reproduction fidelity. To address this limitation, we present a hybrid learning-control architecture that integrates a DS-based IL reference generator with a neuroadaptive controller. Our approach introduces a control action that drives the USV back to the demonstrated path following exogenous disturbances, enabling dynamic human-like reactive alignment—termed \emph{behavioral tracking}. The proposed methodology is validated using the Marine Systems Simulator (MSS) toolbox. Simulation results confirm that the framework generalizes complex maneuvering tasks while substantially improving trajectory tracking fidelity under disturbances compared to alternative control strategies.
\end{abstract}


\end{frontmatter}

\section{Introduction}
\label{sec1}

Collaborative robots have emerged as an option to address the incremental interest in easing human-robot interaction. However, giving a robot a new skill demands many hours of explicit programming by trained personnel. This limitation can be overcome by allowing the robot to learn tasks instead of relying on conventional programming. This paradigm shift is the key to expanding the versatility and range of applications of robotic systems. The maritime industry can take advantage of these advances in robotics to produce autonomous navigation in ships by pilot indications. Some maritime tasks such as docking a ship become repetitive tasks for an experienced pilot; therefore, rather than doing the same tasks multiple times, the human can transfer these skills to the onboard ship computer in a natural way to reproduce complex maneuvers to accomplish the task.

One potential solution for this challenge is Imitation Learning (IL), which can effectively transfer human motion skills to robots. IL, also known as "Learning from Demonstration" or "Programming by Demonstration" \citep{Billard,Calinon09book,Gribovskaya2011}. This approach is highly valuable for tasks that are too complicated to program traditionally. IL is a three-step methodology —demonstration, learning, and reproduction— that allows a robot to acquire new motion skills implicitly by learning from a human-expert, rather than requiring complex programming for every new task. Demonstration data can be collected in multiple ways; kinesthetic  \citep{Kronander2014,Sakr2020}, teleoperation \citep{Havoutis2019,Zhang2018}, or passive observation \citep{Liu2018,Wang2022}. Supervised and unsupervised methods are employed to create a learning-based model from the collected data which represents robot kinematics. Once the model is trained, this creates the appropriate motion for the robot to reproduce the demonstrated task. Complex trajectories are the common skill to transfer; therefore, a set of demonstrated motions is learned, after which a generalization\footnote{Generalization is the ability to apply learned knowledge to new `unfamiliar' situations.} is retrieved for execution by the robot.

IL methodology is widely used to transfer human knowledge of complex, repetitive tasks into robotic and automated systems, such as digging or harvesting in the agricultural field \citep{Lauretti2023,KimChung2025}, needle manipulation and peg transfer in surgical procedures \citep{LiBin2022,Schwaner2021}, human-robot cooperation \citep{Koskinopoulou2016,Sasagawa2020}, and cutting vegetables or handling raw materials for food products \citep{Lioutikov2016,Misimi2018}. However, noisy environments such as the maritime industry can challenge the performance of repetitive tasks; therefore, disturbances must be considered in the motion planning for ships. An IL strategy is proposed in \cite{Piyabhum2023} for navigating an Unmanned Surface Vehicle (USV). More robust strategies have been developed for learning in \cite{BECERRAMORA2025} and for control in \cite{PEILONG}.

The motion of a system can be represented by Dynamical Systems (DS), which is defined in state-space and governed by a set of ordinary differential equations (ODEs). DS is a common method to model complex trajectories and to produce real-time motion from any starting point. Furthermore, the target point of the task can be encoded as a stable attractor by using DS. One of the most used time-dependent DS to learn from demonstrations and generate motions is Dynamic Movement Primitives (DMP) \citep{Ijspeert2013}; nevertheless, some drawbacks such as poor generalization and difficulty dealing with temporal perturbations are common. Therefore, an alternative to guarantee robustness against temporal perturbations is state-dependent DS, which can be modeled using IL methodology. 

A diversity of approaches in DS-based IL have been proposed to learn complex dynamics and ensure stability. In \cite{Khansari2011}, a quadratic Lyapunov function is used to constrain the parameters of the Gaussian Mixture Model (GMM). In \cite{Neumann2013}, a predefined Lyapunov function is used to guarantee asymptotic stability during the learning process via sampling inequality constraints. In \cite{Khansari2014}, the problem is divided into three steps: (i) a Control Lyapunov Function (CLF) is learned from a set of demonstrations; (ii) a supervised or unsupervised learning method is used to learn a non-linear trajectory; and (iii) a constrained optimization problem is solved to ensure stability. An enhanced formulation by \cite{BecAco24} unified the learning of the CLF and non-linear trajectory into a single noise-tolerant constrained optimization problem. This approach was later validated in practice using real ship maneuvering data from a port environment in \cite{BECERRAMORA2025}.
In \cite{Neumann2015}, a Lyapunov function is transformed from the original space to a different space to improve stability and accuracy via diffeomorphic transformations. In \cite{Santos2018}, a non-linear autoregressive polynomial model and a constrained least-square estimator are used to implement an IL methodology, local asymptotic stability is guaranteed. In \cite{Duan2019}, three factors (accuracy, stability and speed) are considered to propose an algorithm based on Extreme Learning Machine that learns from demonstrations. A neural-network-based DS is presented in \cite{Zhang2022} to learn complex motions from demonstrations. Consequently, (un)supervised learning methods are used to learn complex dynamics by means of constrained optimization problems with stability guarantees.

However, successfully reproducing a task requires imitating the underlying motion pattern and guaranteeing convergence to the target, safety- and precision-critical applications demand more. The system must not only recover from perturbations to reach the target, but also accurately track the intended nominal trajectory post-external disturbance. Examples include autonomous Unmanned Surface Vehicles (USVs) navigating narrow channels or robotic manipulators executing high-fidelity welding operations, both of which require adaptive mechanisms to handle exogenous disturbances.

In marine applications, persistent external disturbances—such as ocean waves, currents, and wind—continuously perturb a ship's trajectory, posing a major challenge for precise motion control. To address this issue, we build upon the DS-based IL framework introduced in \cite{BecAco24, BECERRAMORA2025}, in which a constrained optimization problem is solved to simultaneously learn Gaussian Mixture Model (GMM) parameters and an underlying Lyapunov function. Although this approach yields strong nominal performance and theoretical stability guarantees, real-world hardware deployments on Unmanned Surface Vehicles (USVs) reveal a critical limitation: trajectory reproduction fidelity during execution\footnote{Execution refers to hardware deployment, such as running the algorithm on a USV.} degrades significantly under large exogenous perturbations.

Standard DS-based IL methods prioritize global convergence to a target attractor over localized path reproduction under heavy disturbances. Consequently, while the robust controller in \cite{BecAco24} ensures ultimate convergence, it allows state trajectories to deviate substantially from the demonstrated paths before recovering.

To resolve this trade-off between global convergence and localized trajectory fidelity, we propose a novel hybrid learning-control architecture that combines a neuroadaptive controller \citep{Arabi2019} with a learned DS reference. Our main contribution is the introduction of an external control action that actively drives the ship back toward the localized demonstration following an exogenous disturbance, thereby enhancing trajectory recovery while strictly preserving closed-loop stability.

\begin{figure}[hbtp]  
    \centering
    \includegraphics[width=1\columnwidth]{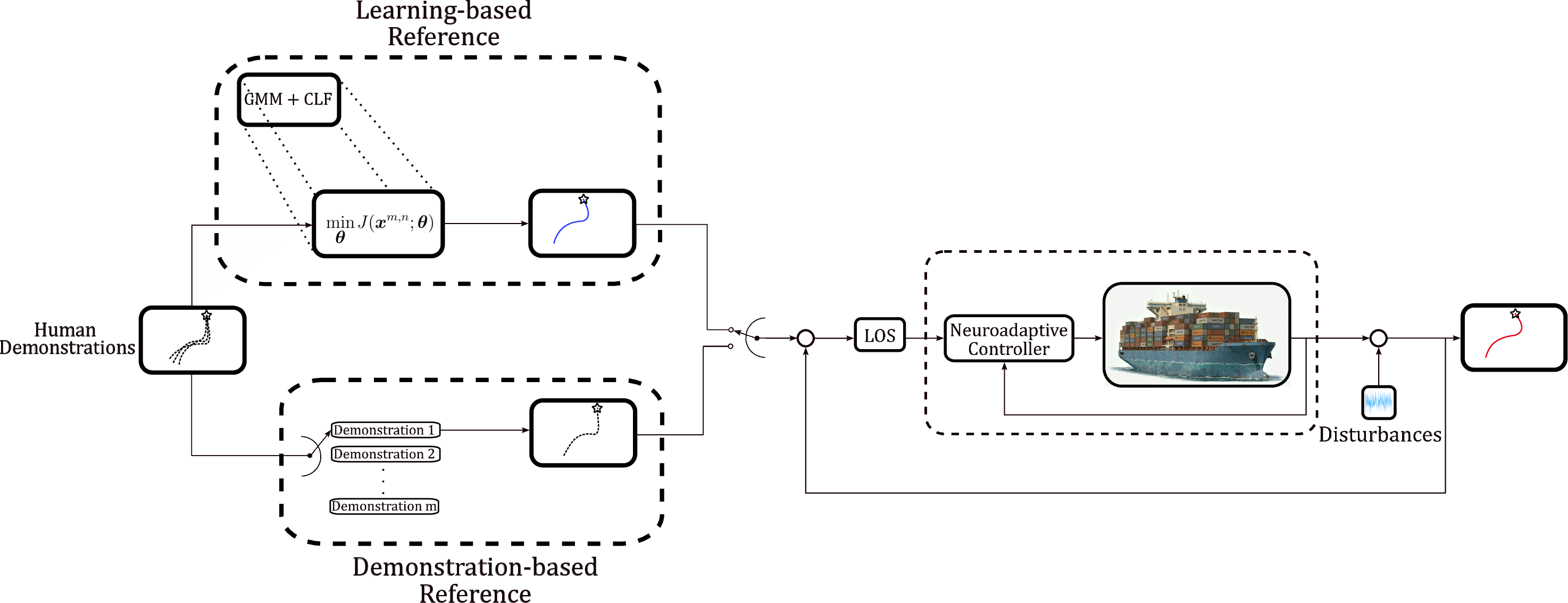}
    \caption{IL architectures comparison: top-left, learning references for human-like responsiveness; bottom-left, demonstration references for classical trajectory tracking control.} 
    \label{fig:BlockD}
\end{figure}

Fig. \ref{fig:BlockD} illustrates the proposed DS-based IL framework. While static references derived directly from human demonstrations lack responsiveness, a common limitation in classical control, learning-based dynamic references capture implicit human behavioral traits, such as adaptive compliance and reactivity. We term this dynamic alignment \emph{behavioral tracking} as it enables the autonomous system to emulate human-like corrective responses under unexpected perturbations. As depicted in Fig. \ref{fig:BlockD}, our framework learns these demonstrations to generate a non-linear target trajectory from any arbitrary geographical coordinate. A Line-of-Sight (LOS) guidance algorithm then processes this reference to generate dynamic heading commands for the neuroadaptive controller.

For the sake of completeness, we perform a comparison integrating different control strategies into the DS-based IL to validate their performance. 

The rest of this paper is structured as follows: Section 2 introduces the background of the learning approach. Section 3 introduces a robust-adaptive approach for high fidelity in imitation learning. In Section 4, the proposed approach is validated in the Marine System Simulator (MSS) toolbox \citep{MSS}. Finally, conclusions are drawn in Section 5.

\smallskip
\noindent \textbf{Notation.} Sets and collections are denoted using braces $\{\cdot\}$. For compactness, bold symbols denote groups of them. The operator $\nabla V(x)$ denotes the gradient of a scalar function $V(x)$ with respect to $x$. $\Vert \cdot \Vert$ stands for vector and induced matrix norms, and $\Vert{}\cdot\Vert{}_\Gamma$ for a weighted matrix norm with $\Gamma$ positive definite. Finally, a continuous and strictly increasing function $\kappa(s) \in \mathcal{K}_\infty$, $\kappa(0) = 0$, and $\lim_{|s| \to \infty} \kappa(|s|) = \infty$.

\section{Robust learning background}
\label{sec2}

Consider a dataset $\mathcal{D}:=\{x^{m,n}, \dot{x}^{m,n}\}_{m=1,n=1}^{M,N}$ with $M$ demonstrations of $N$ samples each, where $x^{m,n}$ represents a position data point and $\dot{x}^{m,n}$  a velocity data point of a platform. Moreover, demonstrations encode point-to-point motions that have the same target state $x^{*}$ for one task. In what follows, we assume that the dataset $\mathcal{D} \subset \mathcal{X}$ consists of points $x \in \mathcal{X}$, where $\mathcal{X} \subset \mathbb{R}^{d}$ is an open set. Hence, a demonstration can be modeled as an autonomous DS:

\begin{equation} \label{eq:sys}
\dot {x} = f(x) + u,
\end{equation}

\noindent where $f:\mathbb{R}^{d} \mapsto \mathbb{R}^{d}$ is a non-linear, continuous, and continuously differentiable function and $u$ the control action of appropriate dimensions; besides, it has a single equilibrium state which can be regarded as the asymptotically stable point attractor $x^{*} = f(x^{*}) = 0$, without loss of generality. (Un)supervised learning methods can be used to learn and reproduce the dynamics of \eqref{eq:sys}. Therefore, the estimated system is represented by

\begin{equation} \label{eq:Esys}
\dot {\hat{x}}=\hat{f}(x; \bm{\theta}) + \hat{u}(x; \bm{\theta}) + \eta(t,x),
\end{equation}

\noindent  where $\hat{f}: \mathbb{R}^{d} \to \mathbb{R}^{d}$ represents the estimated non-linear system dynamics, $\hat{u}(x; \bm{\theta})$ is the estimated control input and $\bm{\theta}$ denotes the parameter vector to be learned. The term $\eta$ stands for a bounded additive disturbance that covers various sources of error, such as measurement inaccuracies, noise-induced imperfections in demonstrations, or external perturbations coming from the environment. The disturbances $\eta$ can negatively affect the learning process, but reproductions can be even more affected; for this reason, the estimated input $\hat u$ is necessary to make corrections within the learning process.

It is worth emphasizing that a first-principles model of the non-linear dynamics is assumed to be unavailable. Consequently, a learning process is required to estimate both the DS and the associated control signal from demonstration data.
On the one hand, an unsupervised learning method such as GMM learns the complex dynamics of the system; on the other hand, a CLF ensures the so-called asymptotic stability of the data-driven dynamics. Therefore, the optimal parameters $\bm{\theta}$ must be learned not only to imitate the demonstrated trajectories, but also to reject disturbances. The remainder of this section briefly summarizes our learning approach, following the core methodology developed in \cite{BecAco24,BECERRAMORA2025}.

%

GMM is a popular method for density approximation and clustering. The $K$ Gaussian kernels are fitted to the dataset $\mathcal{D}$ which are fully described by priors $\pi_k$, means $\mu_k$ and covariance matrices $\Sigma_k$. The mean and covariance matrices are denoted as

\begin{equation}
    \mu_k = 
    \begin{bmatrix}
        \mu_{k}^{x}
        \\
        \mu_{k}^{\dot{x}}
    \end{bmatrix}\textrm{,} \quad
    \Sigma_k = 
    \begin{bmatrix}
        \Sigma_{k}^{x} & \Sigma_{k}^{x\dot{x}}
        \\
        \Sigma_{k}^{\dot{x}x} & \Sigma_{k}^{\dot{x}}
   \end{bmatrix}.
\end{equation}

The probability that each data point $\{x^{m,n},\dot{x}^{m,n}\}$ belongs to a given Gaussian kernel $k$ is defined as the mixture of Gaussian distributions as

\begin{equation}
\mathcal{P}(x^{m,n},\dot{x}^{m,n}; {\bm \theta_k}) = \sum_{k=1}^{K}\pi_k \mathcal{P}({x^{m,n},\dot{x}^{m,n}} | k)
\
    \begin{dcases}
        k \in 1,\ldots,K;
        \\
        n \in 1,\ldots,N;
        \\
        m \in 1,\ldots,M,
    \end{dcases}
\end{equation}

\noindent where $\bm{\theta_k}=\{\pi_k,\mu_{k},\Sigma_{k}\}$ are the GMM parameters, where $\pi_k$ quantifies the Gaussian contribution to the data point, and $\mathcal{P}(x^{m,n},\dot{x}^{m,n}|k)$ is the conditional probability density function calculated with the normal density function for the $k^{th}$ component $\mathcal{N}(x^{m,n},\dot{x}^{m,n}; \mu_{k}, \Sigma_{k})$. The process of optimizing the GMM parameters $\bm{\theta_k}$ is performed by the Expectation-Maximization (EM) algorithm \citep{Bishop}.
Thus, having an optimal $\bm{\theta_k}$ and using Gaussian Mixture Regression (GMR) \citep{Cohn1996}, an estimate of velocity $\hat{f}(x)$ is calculated given the position $x$ as input. This can be computed as


\begin{equation} \label{eq:fest}
    \hat{f}(x; \bm{\theta}) = {\sum_{k=1}^{K} \gamma_k(x) \left(\mu_{k}^{\dot{x}}+\Sigma_{k}^{\dot{x}x}(\Sigma_k^{x})^{-1}({x}-\mu_{k}^{x})\right)},
\end{equation}

\noindent where $\gamma_k(x)$ is a non-linear weighting term to measure the influence of every Gaussian kernel, and the super-indexes are omitted for clarity. A more detailed description of GMM and GMR is presented in \cite{BECERRAMORA2025}. The greater the number of Gaussian components $K$, the better the estimated DS. 

%
%
However, (un)supervised learning methods are prone to instability under noise and perturbations, which can lead the estimated system toward spurious attractors or divergence. To overcome this, non-linear control strategies can be integrated directly into the optimization loop, thereby guaranteeing the stability of the learned estimate as presented in \cite{BecAco24}. In essence, the approach relies on Lyapunov stability theory, which establishes that an autonomous dynamical system is globally asymptotically stable at $x^*$ if there exists a continuously differentiable function $V: \mathbb{R}^d \to \mathbb{R}$ that satisfies three fundamental conditions: positive definiteness, negative definiteness of its time derivative, and $V(x^*) = 0$.
Multiple options to obtain a stable DS refer to utilize learning methods and fulfill the Lyapunov conditions. However, a way to avoid an unstable estimate is to learn the convenient Control Lyapunov Function\footnote{CLFs are essentially a natural generalization of Lyapunov functions for systems with control inputs, see e.g. \cite{Acosta}.} (CLF) and GMM parameters from a set of demonstrations in a single step (see \cite{BecAco24}) for more details).
%
%
The control signal is obtained by analytically solving a constrained optimization problem and using a simplified version of the original Sontag’s universal formula \citep{Sontag} as follows:

\begin{equation}\label{eq:uopt}    
    \hat u(x; \bm{\theta})=   
        -\left(\nabla_x V(x)^{\top} \hat{f}(x; \bm{\theta}) + \rho(|x|)\right) \frac{\nabla_x V(x)}{\Vert \nabla_x V(x)\Vert_2^2},           
\end{equation}

\noindent if $\nabla_x V(x)^{\top} \hat{f}({x}) > -\rho(|x|)$, and 0 otherwise; $V(x)$ represents the actual CLF, and
%
$\rho(|x|) := \rho_0 \big((\nabla_x V(x)^{\top} \hat{f}(x))^{2} + \Vert \nabla_x V(x)^{\top} \Vert_2^4\big)^{1/2}$ with $\rho_0>0$. A suitable CLF to reproduce a broad spectrum of complex DS is the Weighted Sum of $L$ Asymmetric Quadratic Functions (WSAQF) \citep{BecAco24,Khansari2014}. This energy function ensures a unique global minimum at the target point by requiring a set of matrices $P_{l}$ to be positive definite; additionally, the asymmetric shape of this function is influenced by the vectors $\mu_{l}$, with $l=1,...,L$. 


Finally, the optimal parameters to reproduce the DS and ensure its stability are learned using a constrained optimization problem. Merging the GMM and CLF parameters in the collection $\bm \theta := \{\bm \theta_k, P_{l}, \mu_{l}\}$, the optimal parameters $\bm{\theta}$ are learned by minimizing the objective function
\begin{mini}|s|
{\bm \theta}{J({\bm x}^{m,n}; \bm \theta) :=\frac{1}{2 M N} \sum_{m=1}^{M} \sum_{n=1}^{N} | \dot{x}^{m,n}-\dot{\hat x}^{m,n} |^{2}.}
{\label{eq:J}}{}
\end{mini}

This function attempts to reduce the error between the real velocity $\dot{x}$ and its estimate $\dot{\hat{x}}$ while preserving the stability in the DS \citep{BecAco24}. Note that the number of Gaussian kernels $K$ and the number of positive definite matrices $L$ affect the computational time. Therefore, the greater the complexity of the dynamical system, the higher the computational time.    

\begin{remark}
Although the offline learning phase accounts for noise in the demonstration data, an external control action remains necessary to guarantee closed-loop stability during online execution. This requirement motivates Section~\ref{sec:control}, which details the primary contribution of this work.
\end{remark}

\section{Nonlinear Adaptive Control for online tracking execution} \label{sec:control}

Adaptive and neuro-adaptive controllers can achieve high-performance 
system behavior without relying excessively on precise mathematical models. 
Both control strategies effectively mitigate various operational challenges, 
including external noise, system failures, time-varying dynamics, and modeling 
inaccuracies. In this work, we utilize the model reference neuro-adaptive 
control framework proposed in \cite{Arabi2019}. This approach extends 
classical MRAC by closely approximating system uncertainties over a predefined 
compact set, leveraging the universal function approximation property of neural 
networks (NNs). The Model Reference Neuro-adaptive architecture developed in \cite{Arabi2019} successfully ensures the ultimate boundedness of the closed-loop system signals.

In this section, we establish the stability and robustness properties of the 
proposed learning and control framework. For the sake of completeness, we 
provide the rigorous development adapted to our specific formulation following 
the foundational framework in \cite{Arabi2019}. However, we highlight the 
novelty introduced by our integrated learning mechanism, which embeds 
behavioral capabilities into the closed-loop adaptive system. 

Let ${x}_{\scalebox{0.5}{\textit{N}}}$ denote the measurable state vector, living in the compact set $\mathcal{X}_{\scalebox{0.5}{\textit{N}}}\subset\mathcal{X}$, and consider the structure of the  nonlinear system \eqref{eq:Esys} as
\begin{equation} \label{eq:sysA}
\dot {x}_{\scalebox{0.5}{\textit{N}}} = A_{\scalebox{0.5}{\textit{N}}} {x}_{\scalebox{0.5}{\textit{N}}} + B_{\scalebox{0.5}{\textit{N}}} \Lambda \left(\delta(x_{\scalebox{0.5}{\textit{N}}}) + u_{\scalebox{0.5}{\textit{N}}}\right),
\end{equation}

\noindent where $\Lambda$ is the control effectiveness uncertainty, $u_{\scalebox{0.5}{\textit{N}}}$ is the control input, and $\delta(x_{\scalebox{0.5}{\textit{N}}})$ encapsulates the unstructured uncertainty and the approximation error. In this study, ${x}_{\scalebox{0.5}{\textit{N}}}$ and $\dot{x}_{\scalebox{0.5}{\textit{N}}}$ represent the position and velocity of the ship, respectively. The vector $\delta(x_{\scalebox{0.5}{\textit{N}}})$ accounts for the model uncertainties of the ship, the environmental disturbances and the approximation error, such that for $x_{\scalebox{0.5}{\textit{N}}} \in \mathcal{K}$, the whole unstructured uncertainty can be approximated by 
\begin{equation} \label{eq:uncert}
\delta(x_{\scalebox{0.5}{\textit{N}}}) = W_{\scalebox{0.5}{\textit{N}}_{\textit{0}}}^{\top} \Theta_{\textit{0}}(x_{\scalebox{0.5}{\textit{N}}})  + \varepsilon_{\scalebox{0.5}{\textit{N}}} (x_{\scalebox{0.5}{\textit{N}}}),
\end{equation}

\noindent where $W_{\scalebox{0.5}{\textit{N}}_{\textit{0}}}$ is an unknown weighting constant matrix, $\Theta_{\textit{0}}$ are radial basis functions (RBF), and $\varepsilon_{\scalebox{0.5}{\textit{N}}}$ is the approximation error. For the system defined in \eqref{eq:sysA}, the feedback control law is defined as follows

\begin{equation} \label{eq:uN}
u_{\scalebox{0.5}{\textit{N}}} = \underbrace{ -\alpha_1 x_{\scalebox{0.5}{\textit{N}}} + \alpha_2 \hat {x}(t)}_{=:u_n(x_{\scalebox{0.5}{\textit{N}}}, \hat x)}  -\underbrace{\hat{W}_{\scalebox{0.5}{\textit{N}}}^{\top}\Theta(x_{\scalebox{0.5}{\textit{N}}})}_{=:- u_a(x_{\scalebox{0.5}{\textit{N}}},\hat{W}_{\scalebox{0.5}{\textit{N}}})},
\end{equation}

\noindent where $u_n$ and $u_a$ are the nominal and adaptive control laws, respectively. Furthermore, $\alpha_1$ and $\alpha_2$ are the nominal gains, $\hat {x}$ represents the reference system, and  $\hat{W}_{\scalebox{0.5}{\textit{N}}}^{\top} \Theta(x_{\scalebox{0.5}{\textit{N}}})$ is the estimate function, to be defined further. 

On the other hand, a reference model is crucial in this MRAC-based control architecture, which encapsulates the desired closed-loop performance. Let $x_{\scalebox{0.5}{\textit{N}}_{\textit{r}}}$ denote the state of that reference model which is similar to \eqref{eq:sysA}, this is defined as
\begin{equation} \label{eq:sysAr}
\dot {x}_{\scalebox{0.5}{\textit{N}}_{\textit{r}}} 
:= A_{\scalebox{0.5}{\textit{N}}_{\textit{r}}} x_{\scalebox{0.5}{\textit{N}}_{\textit{r}}} + B_{\scalebox{0.5}{\textit{N}}_{\textit{r}}} \hat {x}(t), 
\end{equation}

\noindent where $\hat {x}$ represents the learned reference to be followed. It is worth noting that the presence of uncertainties, as well as the adaptive control law, are ignored in \eqref{eq:sysAr}. Importantly, $\hat {x}$ is the main distinct feature when compared with the standard MRAC architecture. In this work, we learn $\hat {x}$, whereas in the case of Model Reference Neuro-adaptive/MRAC regulators it represents a predefined reference.
The reference model matrices are set by the standard MRAC matching conditions, which are derived equating the steady-state response of the certain part (so-called nominal) closed-loop system from \eqref{eq:sysA} and $u_{n}$ from \eqref{eq:uN}, and hence they are set as follows
\begin{equation} \label{eq:sysArdef}
\dot {x}_{\scalebox{0.5}{\textit{N}}}^{ss} 
= A_{\scalebox{0.5}{\textit{N}}} x_{\scalebox{0.5}{\textit{N}}}^{ss}  + B_{\scalebox{0.5}{\textit{N}}} u_n^{ss} 
\equiv A_{\scalebox{0.5}{\textit{N}}_{\textit{r}}} x_{\scalebox{0.5}{\textit{N}}_{\textit{r}}}^{ss} + B_{\scalebox{0.5}{\textit{N}}_{\textit{r}}} \hat {x}^{ss}  = \dot {x}_{\scalebox{0.5}{\textit{N}}_{\textit{r}}}^{ss}, 
\end{equation}

\noindent where the superscript $^{ss}$ stands for steady state. These impose the two algebraic conditions $A_{\scalebox{0.5}{\textit{N}}_{\textit{r}}} \triangleq A_{\scalebox{0.5}{\textit{N}}} - B_{\scalebox{0.5}{\textit{N}}} \alpha_1$ and $B_{\scalebox{0.5}{\textit{N}}_{\textit{r}}} \triangleq B_{\scalebox{0.5}{\textit{N}}} \alpha_2$; and an additional stability requirement of  $A_{\scalebox{0.5}{\textit{N}}_{\textit{r}}}$ is a Hurwitz matrix, which restricts the feasible set of the control gain $\alpha_{1}$ of \eqref{eq:uN}.

The neuroadaptive approach attempts to minimize the difference between the measurable state and the reference-model state, that is reformulated through the error $e := x_{\scalebox{0.5}{\textit{N}}}-x_{\scalebox{0.5}{\textit{N}}_{\textit{r}}}$; hence, the error dynamics can be obtained from \eqref{eq:sysA}, \eqref{eq:uncert}, \eqref{eq:uN} and \eqref{eq:sysAr} as follows
\begin{equation} 
\begin{split}
\dot {e} = & A_{\scalebox{0.5}{\textit{N}}} x_{\scalebox{0.5}{\textit{N}}} + B_{\scalebox{0.5}{\textit{N}}} \Lambda \left(W_{\scalebox{0.5}{\textit{N}}_{\textit{0}}}^{\top} \Theta_{\textit{0}} + \varepsilon_{\scalebox{0.5}{\textit{N}}}  + u_{\scalebox{0.5}{\textit{N}}} \right) - \left( A_{\scalebox{0.5}{\textit{N}}_{\textit{r}}} x_{\scalebox{0.5}{\textit{N}}_{\textit{r}}} + B_{\scalebox{0.5}{\textit{N}}_{\textit{r}}} \hat {x}(t) \right)\\
= &  A_{\scalebox{0.5}{\textit{N}}} x_{\scalebox{0.5}{\textit{N}}} + B_{\scalebox{0.5}{\textit{N}}} u_n + B_{\scalebox{0.5}{\textit{N}}} \Lambda (I-\Lambda^{-1}) u_n - \left( A_{\scalebox{0.5}{\textit{N}}_{\textit{r}}} x_{\scalebox{0.5}{\textit{N}}_{\textit{r}}} + B_{\scalebox{0.5}{\textit{N}}_{\textit{r}}} \hat {x}(t) \right) \\
& + B_{\scalebox{0.5}{\textit{N}}} \Lambda \left(W_{\scalebox{0.5}{\textit{N}}_{\textit{0}}}^{\top} \Theta_{\textit{0}} + \varepsilon_{\scalebox{0.5}{\textit{N}}} + u_a \right) \\
= &  (A_{\scalebox{0.5}{\textit{N}}} - B_{\scalebox{0.5}{\textit{N}}} \alpha_1) x_{\scalebox{0.5}{\textit{N}}} + B_{\scalebox{0.5}{\textit{N}}} \alpha_2 \hat {x}(t) 
- \left( A_{\scalebox{0.5}{\textit{N}}_{\textit{r}}} x_{\scalebox{0.5}{\textit{N}}_{\textit{r}}} + B_{\scalebox{0.5}{\textit{N}}_{\textit{r}}} \hat {x}(t) \right)\\
& + B_{\scalebox{0.5}{\textit{N}}} \Lambda \left( (I-\Lambda^{-1}) u_n + W_{\scalebox{0.5}{\textit{N}}_{\textit{0}}}^{\top} \Theta_{\textit{0}} + \varepsilon_{\scalebox{0.5}{\textit{N}}}  + u_a \right)  \\
= & A_{\scalebox{0.5}{\textit{N}}_{\textit{r}}} e + B_{\scalebox{0.5}{\textit{N}}} \Lambda \left(W_{\scalebox{0.5}{\textit{N}}}^{\top} \Theta + \varepsilon_{\scalebox{0.5}{\textit{N}}}  + u_a \right) \\
= & A_{\scalebox{0.5}{\textit{N}}_{\textit{r}}} e + B_{\scalebox{0.5}{\textit{N}}} \Lambda \left(W_{\scalebox{0.5}{\textit{N}}}^{\top} \Theta + \varepsilon_{\scalebox{0.5}{\textit{N}}}  - \hat{W}_{\scalebox{0.5}{\textit{N}}}^{\top}\Theta \right) \\
= & A_{ \scalebox{0.5} {\textit{N}}_{\textit{r}} } e - B_{\scalebox{0.5}{\textit{N}}} \Lambda \Tilde{W}_{\scalebox{0.5}{\textit{N}}}^{\top} \Theta(x_{\scalebox{0.5}{\textit{N}}}) + B_{\scalebox{0.5}{\textit{N}}} \Lambda \varepsilon_{\scalebox{0.5}{\textit{N}}} (x_{\scalebox{0.5}{\textit{N}}}),
\end{split}
\label{eq:error}
\end{equation}
\noindent where $\Tilde{W}_{\scalebox{0.5}{\textit{N}}} \triangleq \hat{W}_{\scalebox{0.5}{\textit{N}}} - W_{\scalebox{0.5}{\textit{N}}}$ is the weight estimate error, and we have defined ${W}_{\scalebox{0.5}{\textit{N}}} := [ W_{\scalebox{0.5}{\textit{N}}_{\textit{0}}} \ \vdots \  (I - \Lambda^{-1})^{\top} ] $ and $\Theta(x_{\scalebox{0.5}{\textit{N}}}, \hat x) := [\Theta_{\textit{0}} \ \vdots \ u_n ]$ as the stacked unknown weight matrix and the so-called regressor, respectively. 
Recall that $\dot{\Tilde{W}}_{\scalebox{0.5}{\textit{N}}} = \dot{\hat{W}}_{\scalebox{0.5}{\textit{N}}}$, under the fairly standard assumption that $W_{\scalebox{0.5}{\textit{N}}}$ constant.
Finally, to complete the controller design, the adaptive parameter update law is given by
\begin{equation} \label{eq:uncert_dyn}
\dot{\hat{W}}_{\scalebox{0.5}{\textit{N}}} = \gamma \left( \Theta(x_{\scalebox{0.5}{\textit{N}}})e^{\top} P_{\scalebox{0.5}{\textit{N}}} B_{\scalebox{0.5}{\textit{N}}} - \sigma_{\scalebox{0.5}{\textit{N}}} \hat{W}_{\scalebox{0.5}{\textit{N}}} \right),
\end{equation}

\noindent where $\gamma$ is the learning rate, and the term $\sigma_{\scalebox{0.5}{\textit{N}}}$ is the leakage modification to prevent the estimates from growing unbounded due to drift-causing effects. The positive definite matrix $P_{\scalebox{0.5}{\textit{N}}}$ is the solution of the Lyapunov equation $ A_{\scalebox{0.5}{\textit{N}}_{\textit{r}}}^{\top} P_{\scalebox{0.5}{\textit{N}}} + P_{\scalebox{0.5}{\textit{N}}} A_{\scalebox{0.5}{\textit{N}}_{\textit{r}}} = - R_{\scalebox{0.5}{\textit{N}}}$, for some positive definite matrix $R_{\scalebox{0.5}{\textit{N}}}$. Recall that, a solution always exits, if and only if the matrix $A_{\scalebox{0.5}{\textit{N}}_{\textit{r}}}$ is Hurwitz.


The stability guarantees are established through a Lyapunov analysis. Thus, let define the following positive definite and radially unbounded error function
\begin{equation} \label{eq:Lya}
V(e,\Tilde{W}_{\scalebox{0.5}{\textit{N}}}) := e^{\top} P_{\scalebox{0.5}{\textit{N}}} e + \frac{1}{\gamma} \text{tr}\left[ (\Tilde{W}_{\scalebox{0.5}{\textit{N}}} \Lambda^{1/2})^{\top} (\Tilde{W}_{\scalebox{0.5}{\textit{N}}} \Lambda^{1/2}) \right],
\end{equation}
where $\text{tr}[\cdot]$ stands for the matrix \emph{trace} operator. The time derivative of \eqref{eq:Lya} along the error dynamics \eqref{eq:error} yield
\begin{equation} \label{eq:DLya}
\begin{split}
\dot{V} 
=& - e^{\top} R_{\scalebox{0.5}{\textit{N}}} e + 2 e^{\top} P_{\scalebox{0.5}{\textit{N}}} B_{\scalebox{0.5}{\textit{N}}} \Lambda \left(W_{\scalebox{0.5}{\textit{N}}}^{\top} \Theta + \varepsilon_{\scalebox{0.5}{\textit{N}}}  - \hat{W}_{\scalebox{0.5}{\textit{N}}}^{\top}\Theta \right) \\
&+ \frac{2}{\gamma} \text{tr} \left[ \Lambda^{1/2} \Tilde{W}_{\scalebox{0.5}{\textit{N}}}^{\top} \gamma \left( \Theta e^{\top} P_{\scalebox{0.5}{\textit{N}}} B_{\scalebox{0.5}{\textit{N}}} - \sigma_{\scalebox{0.5}{\textit{N}}} \hat{W}_{\scalebox{0.5}{\textit{N}}} \right) \Lambda^{1/2} \right]  \\
=&  - e^{\top} R_{\scalebox{0.5}{\textit{N}}} e + 2 e^{\top} P_{\scalebox{0.5}{\textit{N}}} B_{\scalebox{0.5}{\textit{N}}} \Lambda \left( \varepsilon_{\scalebox{0.5}{\textit{N}}} - \tilde{W}_{\scalebox{0.5}{\textit{N}}}^{\top}\Theta \right) \\
&+ 2 \ \text{tr} \left[ \Lambda^{1/2} \Tilde{W}_{\scalebox{0.5}{\textit{N}}}^{\top}\Theta e^{\top} P_{\scalebox{0.5}{\textit{N}}} B_{\scalebox{0.5}{\textit{N}}} \Lambda^{1/2} 
-  \frac{\sigma_{\scalebox{0.5}{\textit{N}}} }{\gamma} \Lambda^{1/2} \Tilde{W}_{\scalebox{0.5}{\textit{N}}}^{\top} (\tilde{W}_{\scalebox{0.5}{\textit{N}}}+{W}_{\scalebox{0.5}{\textit{N}}}) \Lambda^{1/2} \right] \\
=&  - e^{\top} R_{\scalebox{0.5}{\textit{N}}} e + 2 e^{\top} P_{\scalebox{0.5}{\textit{N}}} B_{\scalebox{0.5}{\textit{N}}} \Lambda \varepsilon_{\scalebox{0.5}{\textit{N}}}
- 2 \frac{\sigma_{\scalebox{0.5}{\textit{N}}} }{\gamma} \text{tr} \left[\|\tilde{W}_{\scalebox{0.5}{\textit{N}}} \|_{\Lambda}^{2} + {W}_{\scalebox{0.5}{\textit{N}}} \Lambda \Tilde{W}_{\scalebox{0.5}{\textit{N}}}^{\top} \right] \\
\leq &  - e^{\top} R_{\scalebox{0.5}{\textit{N}}} e + 2 \|e\|_{2} \|P_{\scalebox{0.5}{\textit{N}}} B_{\scalebox{0.5}{\textit{N}}} \Lambda\|_{2} \ \bar \varepsilon
- 2 \frac{\sigma_{\scalebox{0.5}{\textit{N}}} }{\gamma} \text{tr} \left[\|\tilde{W}_{\scalebox{0.5}{\textit{N}}} \|_{\Lambda}^{2} + \|{W}_{\scalebox{0.5}{\textit{N}}}\|  \|\Tilde{W}_{\scalebox{0.5}{\textit{N}}}\|_{\Lambda} \right] \\
\leq &  - c_1 \|e\|_{2}^{2} - c_2 \text{tr} \left[\|\tilde{W}_{\scalebox{0.5}{\textit{N}}} \|_{\Lambda}^{2} \right] + \frac{\bar \varepsilon^{2}}{c_1} \|P_{\scalebox{0.5}{\textit{N}}} B_{\scalebox{0.5}{\textit{N}}} \Lambda\|_{2}^{2}
+ \frac{\sigma_{\scalebox{0.5}{\textit{N}}} }{\gamma c_2} \text{tr} \left[\|{W}_{\scalebox{0.5}{\textit{N}}} \|_{\Lambda}^{2} \right],
\end{split}
\end{equation}
where $\bar{\varepsilon}$ is an upper bound of the approximation error defined further,   we have used the properties of the \emph{trace} operator, and the last bound comes from the application of the Young's inequality, with positive constants $c_1$ and $c_2$ that always exist by construction.
\begin{remark}
As a side note, reducing the magnitude of the approximation bound 
$\bar{\varepsilon}$ (e.g., by increasing the number of neurons) and selecting 
a small leakage term coefficient both diminish the effect of the positive terms. 
Therefore, the neural network approximation must hold tightly over $\mathcal{X}$, i.e. $\mathcal{X}_{\scalebox{0.5}{\textit{N}}}\simeq \mathcal{X}$, to minimize the approximation error, implying a trade-off between the convergence rate and the parameter drift rate. 
\end{remark}
The following proposition presents the main stability result, highlighting 
the key difference compared to the standard Model Reference Neuro-adaptive approach, 
and outlines the design guidelines for the controller parameters.

\begin{prop}
Consider the uncertain system dynamics \eqref{eq:sysA}-\eqref{eq:uncert} together with the reference dynamics \eqref{eq:sysAr} and the adaptive state feedback \eqref{eq:uN}-\eqref{eq:uncert_dyn}, with the control gains subject to the following restrictions
\begin{align*}
\alpha_1 &: A_{\scalebox{0.5}{\textit{N}}_{\textit{r}}} = A_{\scalebox{0.5}{\textit{N}}} - B_{\scalebox{0.5}{\textit{N}}} \alpha_1, & \alpha_1 & \ \textnormal{constant and} \ A_{\scalebox{0.5}{\textit{N}}_{\textit{r}}} \ \textnormal{Hurwitz},\\
\alpha_2 &: B_{\scalebox{0.5}{\textit{N}}_{\textit{r}}} = B_{\scalebox{0.5}{\textit{N}}} \alpha_2, & \alpha_2 & \ \textnormal{constant},\\
 P_{\scalebox{0.5}{\textit{N}}} &: A_{\scalebox{0.5}{\textit{N}}_{\textit{r}}}^{\top} P_{\scalebox{0.5}{\textit{N}}} + P_{\scalebox{0.5}{\textit{N}}} A_{\scalebox{0.5}{\textit{N}}_{\textit{r}}} = - R_{\scalebox{0.5}{\textit{N}}}, & P_{\scalebox{0.5}{\textit{N}}}, R_{\scalebox{0.5}{\textit{N}}} & \ \textnormal{constant positive definite}, \\
 \gamma, \sigma_{\scalebox{0.5}{\textit{N}}} &: \ \textnormal{positive constants}.
\end{align*}
Under the assumption of $\varepsilon_{\scalebox{0.5}{\textit{N}}}, {W}_{\scalebox{0.5}{\textit{N}}} \in \mathcal{L_{\infty}}$, with $\|\varepsilon_{\scalebox{0.5}{\textit{N}}} (x_{\scalebox{0.5}{\textit{N}}})\|_2 \leqslant \Bar{\varepsilon}$, $ x_{\scalebox{0.5}{\textit{N}}} \in \mathcal{X}_{\scalebox{0.5}{\textit{N}}}$, the errors $e$ and $\tilde{W}_{\scalebox{0.5}{\textit{N}}}$ are uniformly ultimately bounded, with an estimate of the ultimate bound given by $d/\min(c_1,c_2)$,
$d :=\left(\bar \varepsilon^{2} \|P_{\scalebox{0.5}{\textit{N}}} B_{\scalebox{0.5}{\textit{N}}} \Lambda\|_{2}^{2} + (\sigma_{\scalebox{0.5}{\textit{N}}}/\gamma) \text{tr} \left[\|{W}_{\scalebox{0.5}{\textit{N}}} \|_{\Lambda}^{2} \right]\right)$.
\end{prop}
\begin{proof}
Define the stacked error as $z := \text{col}(e, \tilde{W}_{\scalebox{0.5}{\textit{N}}})$. First, note that while the plant dynamics are non-autonomous, the error dynamics 
\eqref{eq:error} are autonomous because the $\hat{x}(t)$ terms cancel out by 
construction. Furthermore, the Lyapunov function \eqref{eq:Lya} satisfies 
$\underline{\kappa}(\|z\|) \leq V(z) \leq \overline{\kappa}(\|z\|)$ for some 
$\underline{\kappa}, \overline{\kappa} \in \mathcal{K}_{\infty}$. Under the assumption of the proposition and given the derivative estimate in  \eqref{eq:DLya}, the boundedness of $z$ directly follows. 
However, establishing overall system stability requires guaranteeing that the 
state estimate $\hat{x}(t)$ remains bounded. According to Proposition 1 in 
\cite{BecAco24}, the closed-loop learning dynamics \eqref{eq:Esys}--\eqref{eq:uopt} 
guarantee an ultimate bound for $\hat{x}(t)$, even in the presence of noisy data. 
Consequently, the simultaneous boundedness of $e$ and $\hat{x}(t)$ implies that 
both $x_{\scalebox{0.5}{\textit{N}}}$ and $x_{\scalebox{0.5}{\textit{N}}_{\textit{r}}}$ are bounded. Additionally, since $W_N \in \mathcal{L}_{\infty}$ 
by assumption, the boundedness of $\tilde{W}_{\scalebox{0.5}{\textit{N}}}$ ensures that $\hat{W}_{\scalebox{0.5}{\textit{N}}}$ is also bounded.

Finally, it is straightforward to see from \eqref{eq:DLya} that there exist constants $c := \min(c_1,c_2)$ and $d$ such that $\dot V \leq -c V + d$, regardless of the initial conditions, and therefore the ultimate bound result also follows.
\end{proof}

The algorithm \ref{alg} summarizes the neuroadaptive approach. This contributes to a simple and efficient way to estimate the non-linear trajectories and to ensure the stability after large perturbations. Choosing as the reference system the trajectories coming from the optimal parameters $\bm{\theta}$.

\begin{algorithm}
	\caption{Neuroadaptive Controller for DS-based IL}    
	\begin{algorithmic}[1]  
        \Procedure {DS-basedIL}{$\mathcal{D}$, $\rho_0$, $K$, $L$, $M$, $N$} \Comment{Details in \citep{BECERRAMORA2025}}
        \State Initialize the parameters $\bm{\theta}_k$, $P_l$ and $\mu_l$
        \State Calculate the estimated system $\dot{\hat{x}}$ by using \eqref{eq:fest} and \eqref{eq:uopt}         
	    \State Learn the optimal parameters from \eqref{eq:J} 
        \EndProcedure       
        \Statex
        \Procedure {NeuroadaptiveApproach}{$\gamma$, $\Lambda$, $\alpha_1$, $\alpha_2$, $\sigma_{\scalebox{0.5}{\textit{N}}}$}
        \State Construct the compact set $\mathcal{X}_{\scalebox{0.5}{\textit{N}}}:=\{x{\scalebox{0.5}{\textit{N}}}, \dot{x}{\scalebox{0.5}{\textit{N}}}\}$
        \State Define the number of RBF in $\Theta$
        \State Calculate the reference system $\dot {x}_{ \scalebox{0.5} {\textit{N}}_{\textit{r}} }$ by using \eqref{eq:sysAr} and \eqref{eq:Esys}
        \State Compute the control law $u_{\scalebox{0.5}{\textit{N}}}$ from \eqref{eq:uN}, \eqref{eq:Esys} and \eqref{eq:uncert_dyn}
        \State Calculate the corrected system $\dot {x}_{\scalebox{0.5}{\textit{N}}}$ by using \eqref{eq:sysA}
        \EndProcedure
	\end{algorithmic} 	    
    \label{alg}
\end{algorithm}

\section{Validation with realistic simulator}

The DS-based IL method presented in this study was validated in the Marine System Simulator (MSS) toolbox \citep{MSS} with a numerical model of a container ship proposed in \cite{Nomoto}. The code and datasets used for this study are freely available at \cite{DSbasedIL}. This model is widely-used as a benchmark for testing the efficacy of control algorithms in marine engineering. A laptop equipped with an Intel\textsuperscript{\textregistered} Core i9-13900H 2.6 GHz CPU and 32 Gb RAM was used to run these simulations. The main features of the container ship are listed in Table~\ref{tb:ShipFeatures}. 

\begin{table}[ht]
\begin{minipage}[b]{0.4\linewidth}
\centering
\begin{tabular}{l c c c} 
    \toprule
    {\textbf{Dimensional parameters}}          & {\textbf{Values}} \\                                                         \midrule                        
    Length             & 175.00 \si{\meter}   \\
    Beam          & 25.40 \si{\meter}\\
    Displaced volume       & 21.22 \si{\cubic\meter}\\
    Fore draft       & 8.00 \si{\meter}\\
    Aft draft       & 9.00 \si{\meter}\\
    Aspect ratio       & 1.82 \\
    Rudder area               & 33.04 \si{\meter\squared}\\
    Propeller diameter              & 6.53 \si{\meter}\\
    \bottomrule
\end{tabular}
\caption{Main features of the container ship model \citep{fossen_guidance}.}
\label{tb:ShipFeatures}
\end{minipage}\hfill
\begin{minipage}[b]{0.4\linewidth}
\centering
\includegraphics[width=1.1\columnwidth]{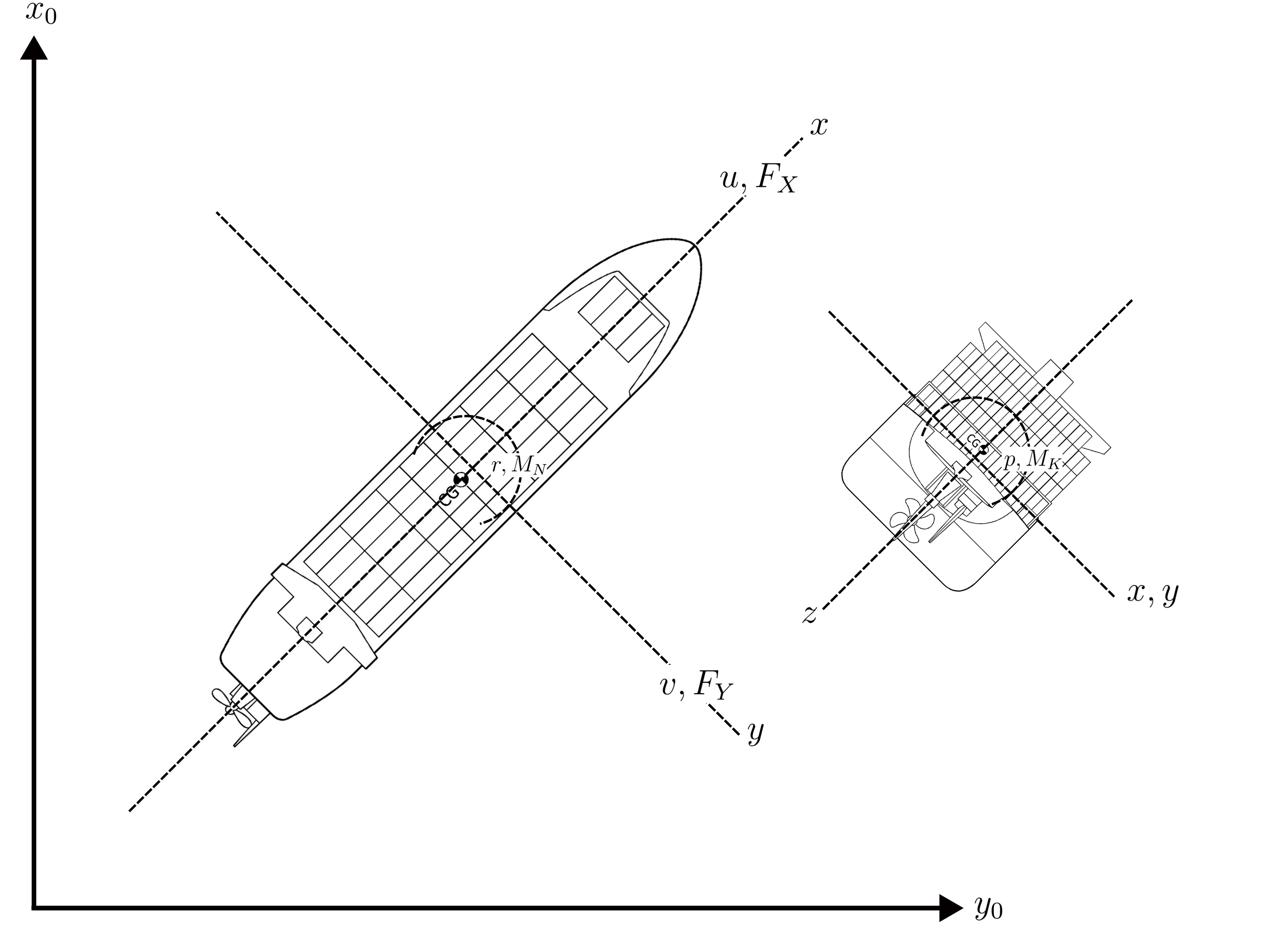}
\captionof{figure}{Coordinate system for the container ship.}
\label{fig:coorSys}
\end{minipage}
\end{table}

The nonlinear equations of motion for the container ship are represented by the surge, sway, and pitch forces, as well as the roll, yaw, and heave moments; nevertheless, the effect of pitch and heave is neglected to obtain a simplified model of the hydrodynamic forces $F$ and moments $M$ (see Fig. \ref{fig:coorSys}). These are described as follows:

\begin{equation} \label{eq:ForMom}
\begin{split}
F_X &= (m + m_x)\dot{u} - (m + m_y)vr \\
F_Y &= (m + m_y)\dot{v} + (m + m_x)ur + m_y \alpha_y \dot{r} - m_y l_y \dot{p} \\
M_K &= (I_x + J_x)\dot{p} - m_y l_y \dot{v} - m_x l_x ur + WGMp \\
M_N &= (I_z + J_z)\dot{r} + m_y \alpha_y \dot{v} + F_Y x_G
\end{split}
\end{equation}

\noindent where $F_X$ and $F_Y$ denote the hydrodynamic surge and sway forces, respectively. $m$ is the mass of the ship, $m_x$ and $m_y$ are the added mass in the $x$ and $y$ direction. $\alpha_y$ denotes the $x$-coordinates of the center of $m_y$, and $l_y$ and $l_x$ the $z$-coordinates of the center of $m_y$ and $m_x$, respectively. $M_K$ is the hydrodynamic roll moment about the center of gravity (CG) and $M_N$ is the hydrodynamic yaw moment about the midship. $I_x$, $I_z$, $J_x$ and $J_z$ denote the moment of inertia and the added moment of inertia about the $x$ and $z$ axes, respectively. $x_G$ is the distance from the CG, $W$ is the weight of the ship, and $GM$ is the metacentric height. A complete analysis of this system is detailed in \cite{Nomoto}.

As mentioned above, the DS-based IL is composed of two steps: a learning process to mimic the human experience and a control strategy to ensure stability. Experiments are developed in the following sections.

\subsection{Learning complex trajectories}

The MSS was initially used to construct sets of complex trajectories using the nonlinear model of the container ship. Three different sets of non-linear trajectories were simulated. Each set likewise contains 3 similar trajectories or demonstrations $M$ (i.e. positions are geographically close) with one common target (e.g. origin of the coordinate reference system). The simulated trajectories cover distances between 2.59 $\si{\nauticalmile}$ and 3.02 $\si{\nauticalmile}$. 

The position and velocity of the ship data are recorded to be used in the learning process. Since the recorded data-points $N$ are massive ($5000$), a pre-processing stage is required to reduce the learning time while preserving their integrity. The number of Gaussian Kernels $K$ and the number of asymmetric quadratic functions $l$ are 5 and 2, respectively. 

The parameter $\bm{\theta_k}$ is initialized with the k-means algorithm and optimized with the EM algorithm to produce an estimated velocity $\hat{f}(x; \bm{\theta})$ (see Eq. \eqref{eq:fest}). The parameters $P_l$ and $\mu_l$ are initialized with $\begin{bmatrix} 1 & 0\\ 0 & 1 \end{bmatrix}$ and $[0,0]^{\top}$, respectively. Additionally, the parameter $\rho_0$ is tuned to $\num{3e-9}$. The parameters $\bm{\theta}$ are optimized by Eq. \eqref{eq:J}. Once these optimal parameters $\bm{\theta}$ defined in Section 2 are learned, the estimated trajectories can be reproduced from any geographical point relatively close to the demonstrations. 

\begin{figure}[hbtp]  
    \centering
    \includegraphics[width=1\columnwidth]{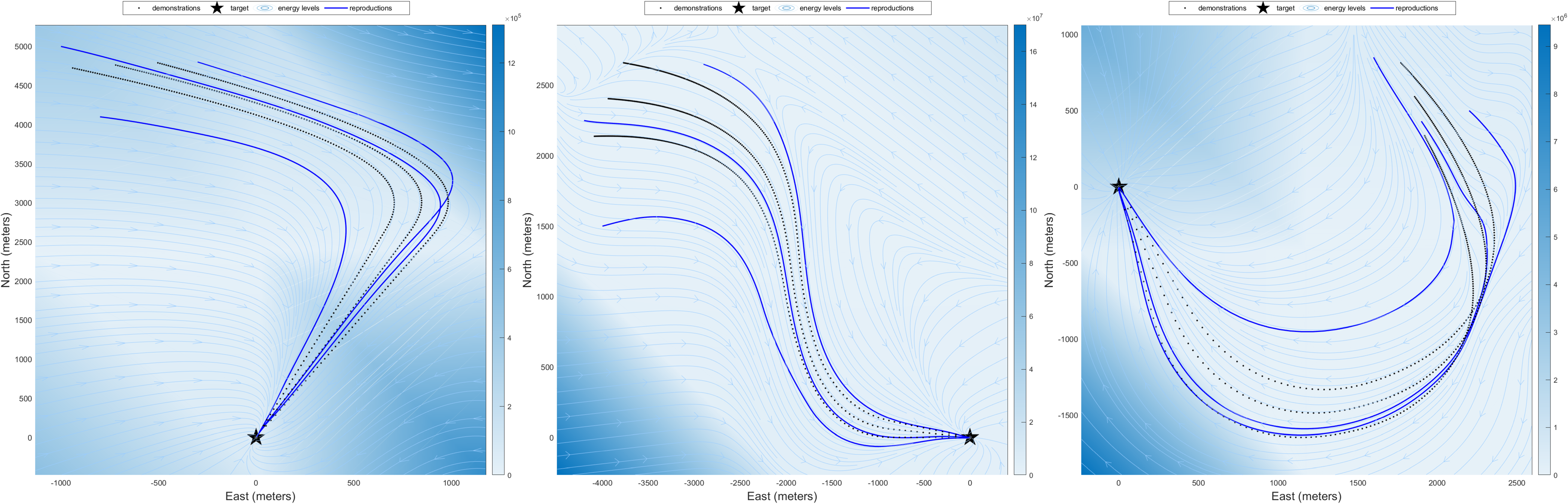}
    \caption{Demonstrated trajectories (black dots) and estimated trajectories (blue lines).}
    \label{fig:learningRef}
\end{figure}

In this context, DS-based IL relies on yielding estimates from a set of demonstrations, thus generalizing the learning of the trajectory. Therefore, if the ship changes its position with respect to recorded data, no human intervention is necessary to re-planning. Demonstrations and estimates are illustrated in Fig. \ref{fig:learningRef}. The black dots, the dark blue lines, and the light blue lines represent the demonstrations performed by a human-expert, the estimates of such demonstrations, and the streamlines, respectively. The background color depicts the rate of change of the energy-like function: the lighter the color, the higher the dissipation. Note that the estimates start from random positions and converge to the unique target; furthermore, their trajectory's shape resembles the demonstrations shape (i.e. generalization capacity).

A path planning algorithm requires a set of way-points to follow a trajectory and subsequently reach a target. These way-points can be recorded within the ship to accomplish a maritime route; however, every time the ship is positioned in a different geographical point, this will have to re-planning the trajectory or will have to navigate to the recorded trajectory, which is inefficient. Hence, a DS-based IL methodology is suitable to avoid such drawbacks. 

\begin{figure}[hbtp]    
    \centering
    \begin{subfigure}[ht]{0.25\textwidth}        
        \centering
        \includegraphics[width=1\columnwidth]{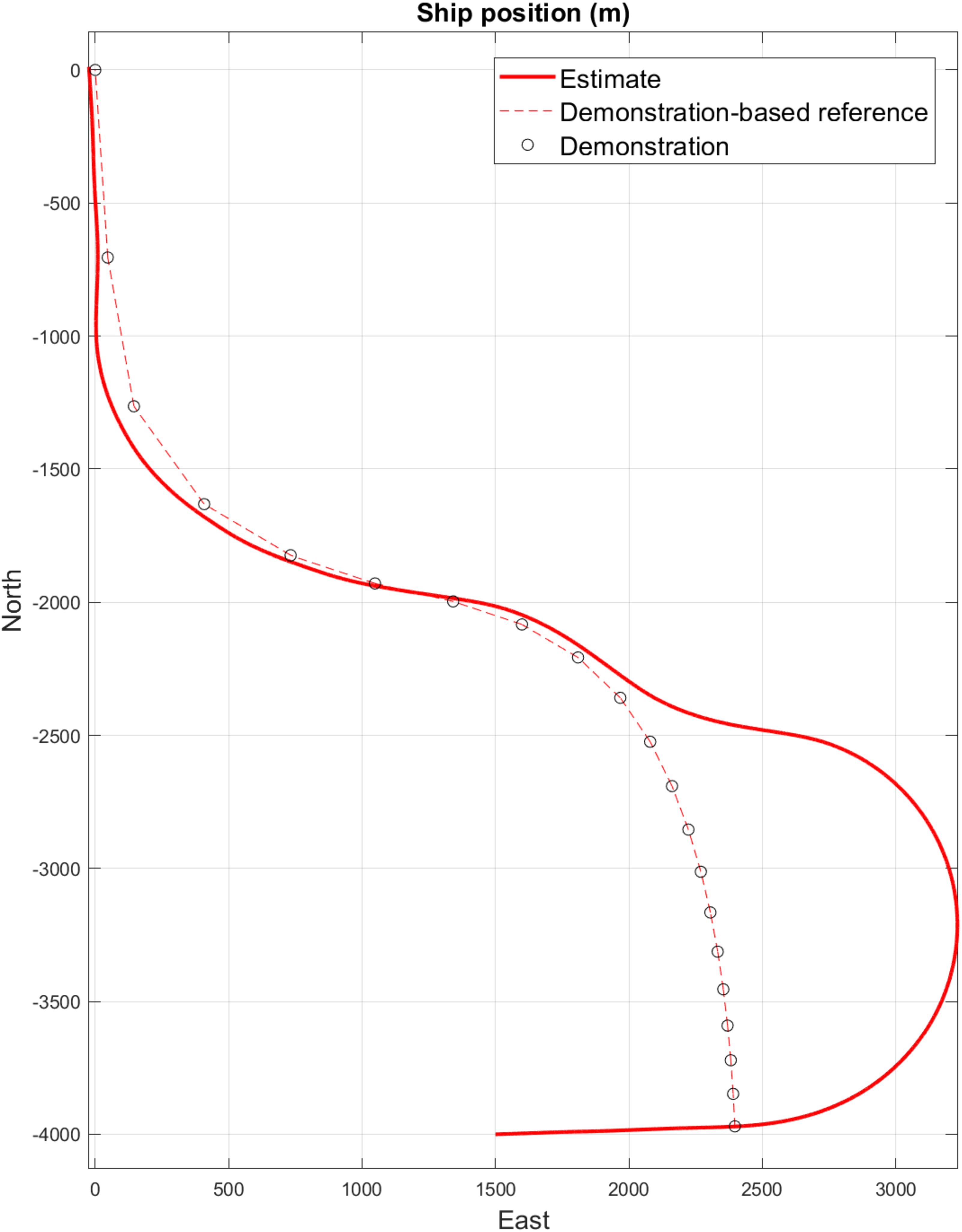}    
        \caption{}
        \label{fig:RefA}
    \end{subfigure}
    \begin{subfigure}[ht]{0.25\textwidth}        
        \centering
        \includegraphics[width=1\columnwidth]{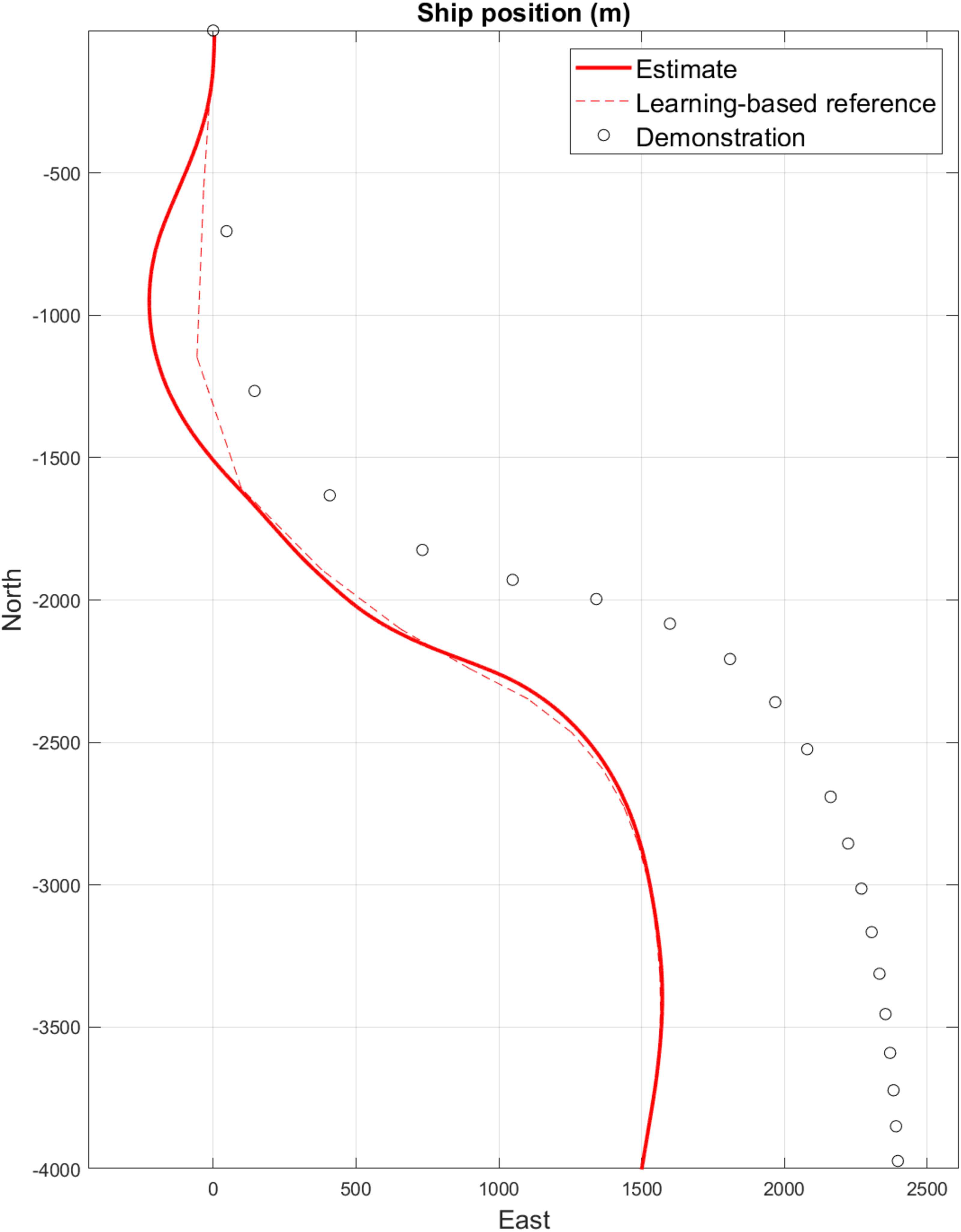}
        \caption{}
        \label{fig:RefB}
    \end{subfigure}
\caption{Comparison of references. Demonstration-based reference (a) and learning-based reference (b).}
\label{fig:References}
\end{figure}

A comparison between two types of reference is shown in Fig. \ref{fig:References}. The initial condition for the container ship in both scenarios is exactly the same, as well as their targets are very similar. However, two downsides can be observed in Fig. \ref{fig:RefA}: (i) the ship must first reach the demonstration to execute the trajectory and (ii) the ship requires more effort to follow the reference. Unlike these downsides, the ship in Fig. \ref{fig:RefB} follows the learned trajectory smoothly.  

\subsection{Trajectory tracking via heading control}

The model of the container ship is capable of following a trajectory autonomously if a collection of way-points is provided and a control strategy is employed to regulate its heading. The way-points are drawn from the learning-based reference presented in the previous section. 

\begin{figure}[h!]  
    \centering
    \includegraphics[width=1\columnwidth]{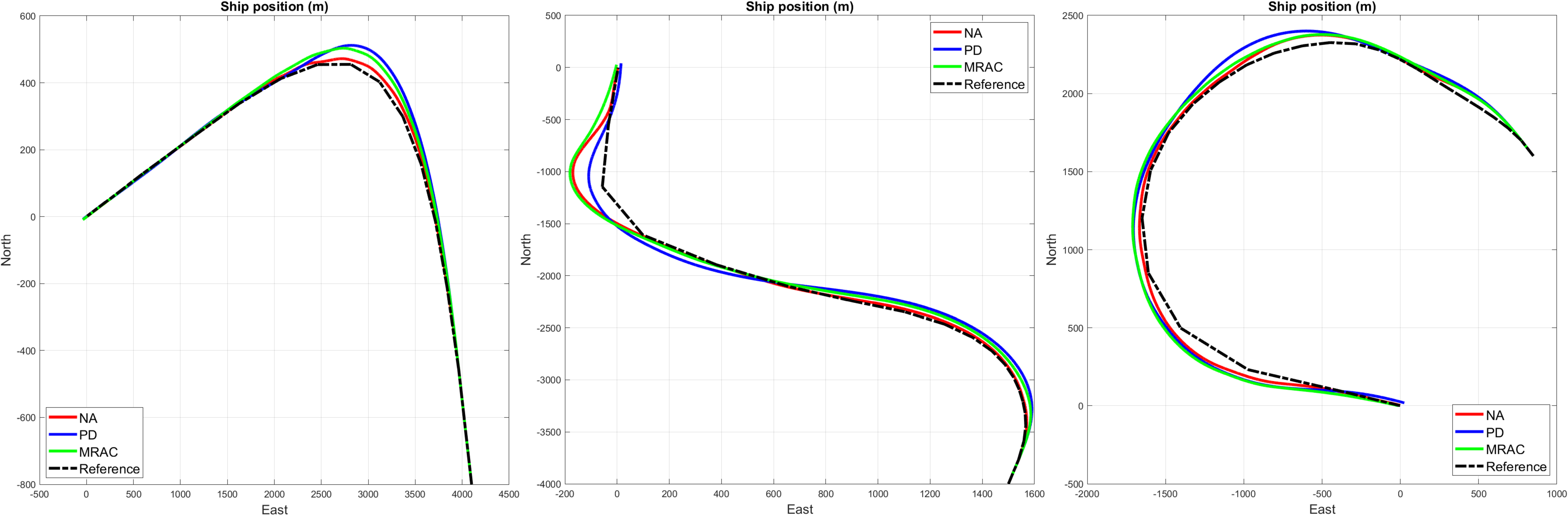}
    \caption{Heading tracking.}
    \label{fig:Tracking}
\end{figure}

The initial conditions for the surge, sway, and yaw velocity are 8 m/s, 0 m/s, and 0 deg/s, respectively. The shaft velocity is 70 rpm (constant) and the angle of the rudder (rad) is used as the control input. The control algorithm to guarantee the stability in the ship is Neuroadaptive control (NA). However, a Model Reference Adaptive Control (MRAC) and Proportional-Derivative control (PD) are used for comparison. 

\begin{table}[h!]
\centering
\footnotesize
\caption{Trajectory tracking error with ideal conditions (in Nautical miles)}
\label{tb:SEA}
\begin{tabular}{l c c c} 
    \toprule
    {\textbf{Type of trajectory}}          &\multicolumn{3}{c}{\textbf{Controllers}} \\
    \cmidrule(lr){2-4}
                           & \text{NA}  & \text{MRAC}  & \text{PD} \\
    \midrule                        
    A          & 0.0111 & 0.0289 & 0.0255 \\
    B          & 0.0368 & 0.0561 & 0.0575 \\
    C          & 0.0490 & 0.0913 & 0.0837 \\    
    \bottomrule
\end{tabular}
\end{table}

The fundamental parameters of the NA control are defined as follows: The number of neurons $\Theta$ was set to 18, the adaptive gain $\gamma$ was set to 50, the control effectiveness uncertainty $\Lambda$ was set to 1, the leakage term $\sigma_{\scalebox{0.5}{\textit{N}}}$ was set to 0.1, the nominal gains $\alpha_1$ and $\alpha_2$ were set to $[3.16,4.04]$ and 3.16, respectively, the solution of the Lyapunov equation $P{\scalebox{0.5}{\textit{N}}}$ was $\begin{bmatrix} 1.15 & 0.16\\ 0.16 & 0.16 \end{bmatrix}$ and the state and input matrices ($A_{\scalebox{0.5}{\textit{N}}}$, $B_{\scalebox{0.5}{\textit{N}}}$) were set to $\begin{bmatrix} 0 & 1\\ 0 & 0 \end{bmatrix}$ and $[0,1]^{\top}$, respectively. In the other hand, the MRAC utilized the first-order Nomoto model as its reference model (see Appendix A for a more detailed explanation).

The trajectory tracking is depicted in Fig. \ref{fig:Tracking}. The red line, blue line, and green line represent the trajectory followed by the ship using NA control, MRAC and PD control, respectively. The black dashed line is the learning-based reference. The similarity among the different control strategies is qualitatively evident. For a more rigorous assessment, the Swept Error Area (SAE) metric \citep{Khansari2014} is utilized to measure the inaccuracy between the reference and the estimated trajectory of each control algorithm. A quantitative comparison is detailed in Table~\ref{tb:SEA}. 

\begin{figure}[h!]  
    \centering
    \includegraphics[width=1\columnwidth]{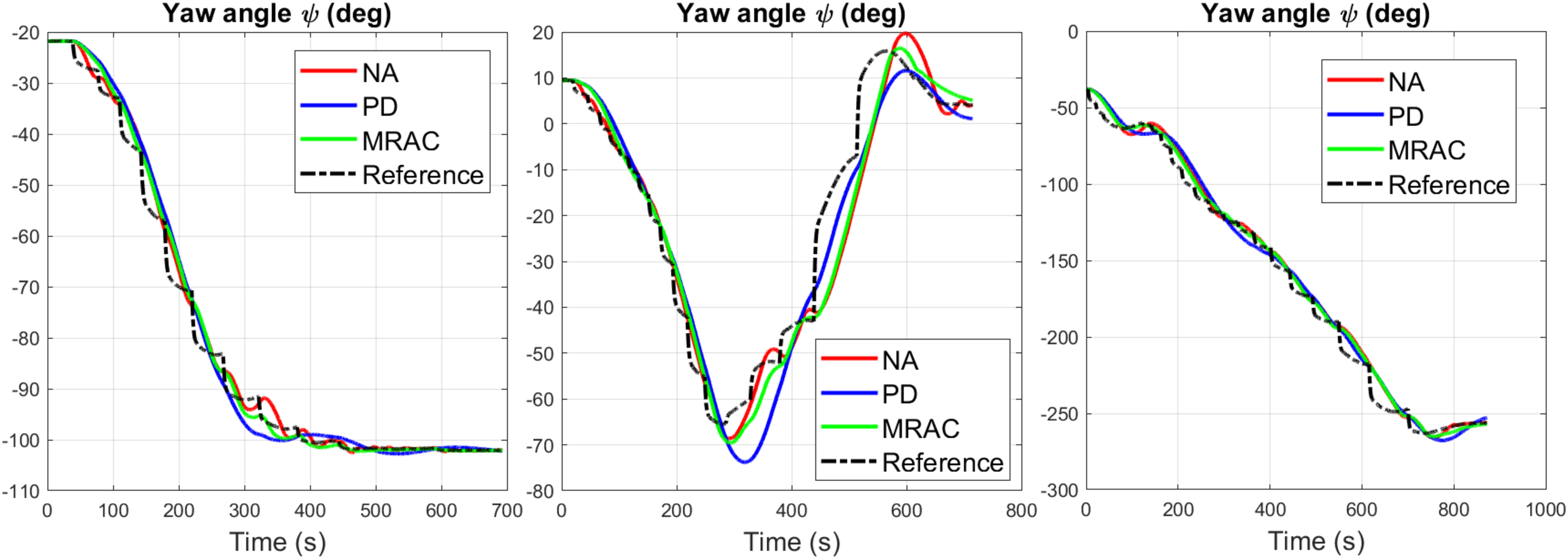}
    \caption{Heading tracking.}
    \label{fig:Heading}
\end{figure}

SEA results are given in squared nautical miles ($\si{\nauticalmile\squared}$); the smaller the area, the larger the precision of the estimate. Note that the NA approach provides the best results among controllers. The heading tracking for every trajectory is shown in Fig. \ref{fig:Heading}. The red line, blue line, and green line represent the heading performed by the ship using NA control, MRAC and PD control, respectively. The black dashed line is the desired heading. Note that the desired heading located in the middle of Fig. \ref{fig:Heading} is more difficult to track due to the quick changes in the direction of the ship.

\subsection{Trajectory tracking with disturbances}

IL is  not sufficient to deal with maritime disturbances due to the inability to ensure system stability; hence, a more sophisticated methodology is required as DS-based IL. The robustness of this methodology is validated by adding Gaussian noise to the system measures. A comparison among NA, MRAC and PD to follow a complex trajectory with noise is presented in Fig \ref{fig:TrackingN}. 

\begin{figure}[hbtp]  
    \centering
    \includegraphics[width=1\columnwidth]{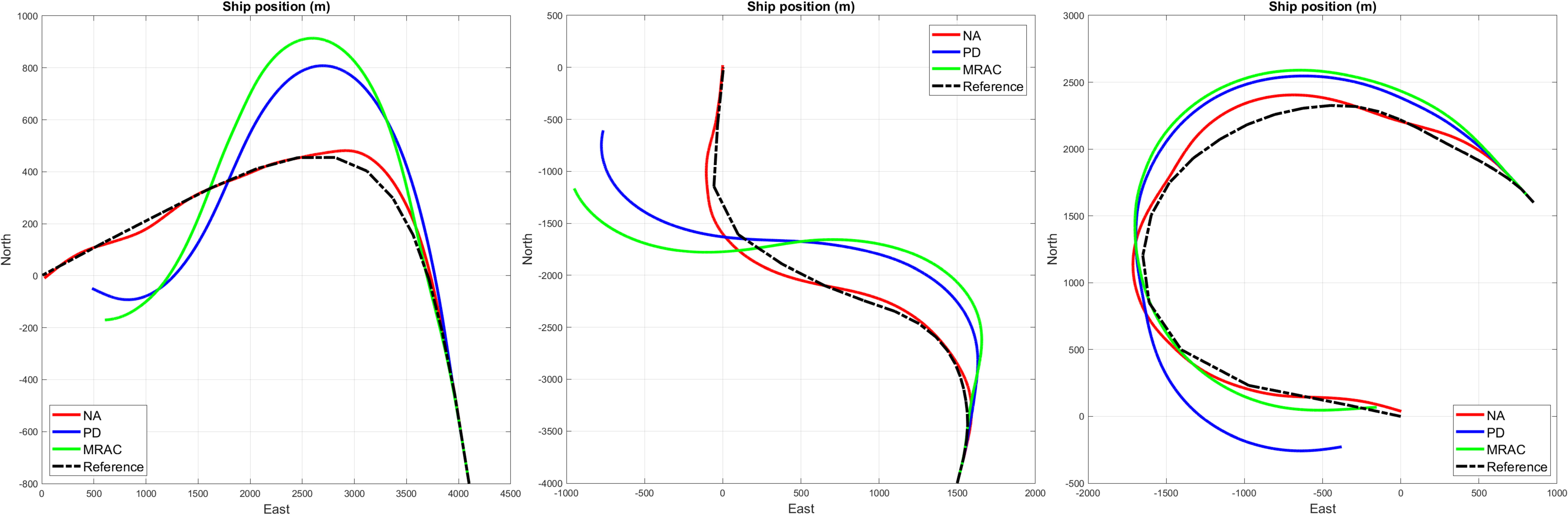}
    \caption{Heading tracking with noisy observations.}
    \label{fig:TrackingN}
\end{figure}

The DS-based IL with a NA approach considerably outperforms PD and MRAC approaches to follow a non-linear trajectory with maritime disturbances. Note that the red line (NA approach) is the closest one to the learning-based reference despite the induced noise. Additionally, heading tracking for every trajectory corroborates this behavior (see Fig. \ref{fig:HeadingN}), as well as SEA results presented in Table~\ref{tb:SEA_N}.

\begin{figure}[hbtp]  
    \centering
    \includegraphics[width=1\columnwidth]{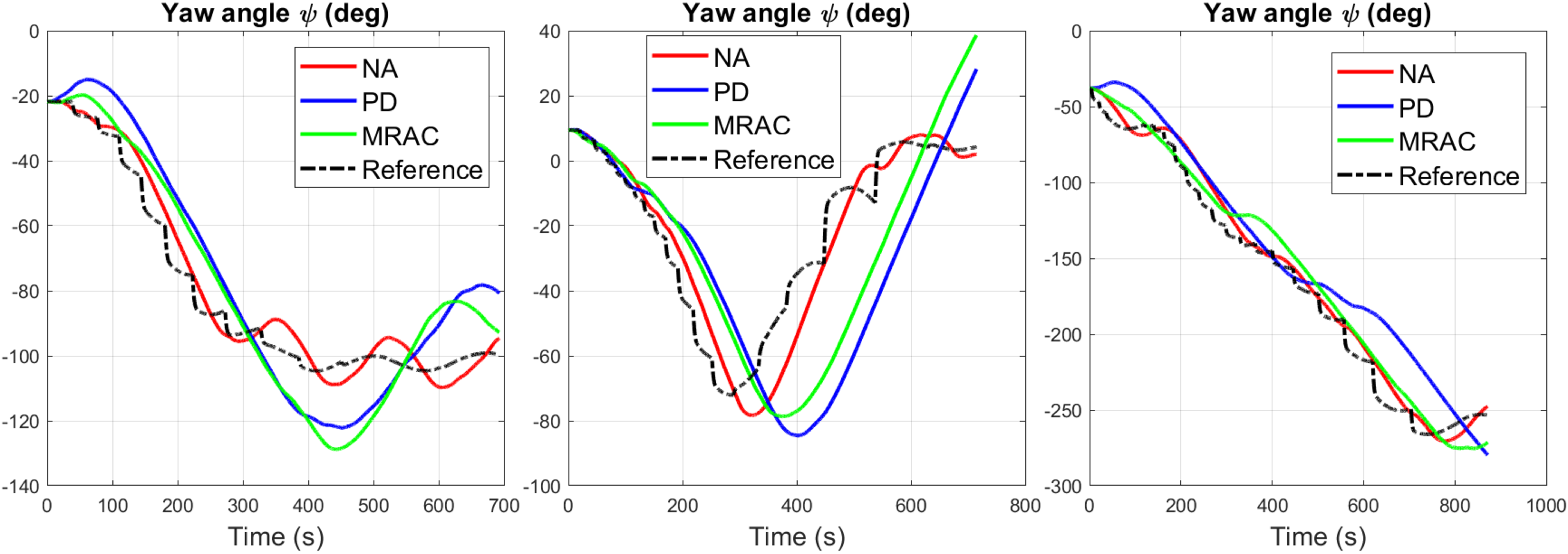}
    \caption{Heading tracking with noisy observations.}
    \label{fig:HeadingN}
\end{figure}

Although different control strategies for a DS-based IL methodology present similar results as can be observed in the previous section, the NA approach offers a reasonable advantage over the other two approaches (MRAC and PD). Furthermore, the trajectory (with noise) tracking error is much larger in the MRAC and PD approaches than in the NA approach (see Table~\ref{tb:SEA_N}).

\begin{table}[ht!]
\centering
\footnotesize
\caption{Trajectory tracking error with noisy observations (in Nautical miles)}
\label{tb:SEA_N}
\begin{tabular}{l c c c} 
    \toprule
    {\textbf{Type of trajectory}}          &\multicolumn{3}{c}{\textbf{Controllers}} \\
    \cmidrule(lr){2-4}
                           & \text{NA}  & \text{MRAC}  & \text{PD} \\
    \midrule                        
    A          & 0.0231 & 0.1561 & 0.1839 \\
    B          & 0.0530 & 0.4268 & 0.4386 \\
    C          & 0.0869 & 0.2552 & 0.2276 \\    
    \bottomrule
\end{tabular}
\end{table}

\section{Conclusions}
This study presented a novel Dynamical System-based Imitation Learning (DS-based IL) framework integrated with neuroadaptive control to enable dynamic, high-fidelity trajectory tracking for autonomous ships under severe marine disturbances. The proposed architecture addresses the critical trade-off between global convergence and localized path fidelity through a two-phase design: (i) generalizing human expert maneuvers into dynamic, state-dependent reference trajectories, and (ii) synthesizing an auxiliary neuroadaptive control action that enforces trajectory recovery while strictly preserving closed-loop stability. Extensive simulations across diverse initial conditions and three distinct trajectory scenarios confirmed the framework's superior generalization compared to static demonstration-based references. Furthermore, comparative evaluations demonstrated that the neuroadaptive scheme yields superior tracking fidelity over standard control baselines in both low- and high-disturbance regimes.

While the proposed architecture exhibits strong robustness to system uncertainties, achieving optimal closed-loop performance requires systematic tuning of key parameters, including neuron density, adaptive gains, leakage terms, and hyperparameters. In practical deployments, this framework mitigates pilot workload during repetitive maneuvering tasks by capturing implicit human compliance and reactivity. Consequently, it provides a viable baseline for advanced decision-support systems on conventional vessels or primary guidance-and-control architectures for Unmanned Surface Vehicles (USVs).

Future work will focus on experimental deployment on a physical USV platform to evaluate real-world hydrodynamic effects. Additionally, an online obstacle avoidance layer will be integrated into the dynamic reference generator to handle dynamic threats and complex marine traffic scenarios.

\section*{Acknowledgements}
This publication is part of the Project PUSHME-TUGS Ref. PID2024-158023OB-C21 which is funded by MICIU/AEI/10.13039/501100011033 and by ERDF/EU. This work was also supported by the Emerging Research Group Multi-Robot and Control Systems \href{https://investigacion.us.es/sisius/sis_depgrupos.php?ct=&cs=&seltext=TEP-995&selfield=CodPAI}{(\textbf{MACS})} and by the National Program for Doctoral Formation (Minciencias-Colombia, 885-2020). 

\appendix
\section{Nomoto model for MRAC}

The first-order Nomoto model is used as a reference model for MRAC. This is described as
\begin{equation*}
\dot{r} + \frac{1}{T}r = \frac{G}{T} u_{\delta},
\end{equation*}

\noindent where $r$ is the derivative of the heading $(\psi)$, $T$ is the time constant, $G$ is the control gain, and $u_{\delta}$ is the rudder. Hence, assuming that the parameters $T$ and $G$ are unknown, the system can be in the form \eqref{eq:sysA} as

\begin{equation*}
\begin{bmatrix} \dot{x}_1\\ \dot{x}_2 \end{bmatrix} = \begin{bmatrix} 0 & 1\\ 0 & 0 \end{bmatrix} \begin{bmatrix} {x}_1\\ {x}_2 \end{bmatrix} + \begin{bmatrix} 0\\ 1 \end{bmatrix} \frac{G}{T} \left[ u - \frac{1}{G} {x}_2 \right],
\end{equation*}

\noindent where $x_1 = \psi$, $x_2 = r$ and $u=u_\delta$. Note that the state and input matrices are exactly the same as those used in neuroadaptive approach; nevertheless, uncertainty was structured according to the Nomoto model.

\bibliographystyle{plainnat}
\bibliography{ref.bib}

\end{document}